\documentclass{article}
\usepackage{amsmath, amssymb, amsthm, braket, latexsym, mathtools}
\usepackage{subfiles,subcaption}
\usepackage{bm,array,arydshln}
\usepackage{graphicx}

\newtheorem{theorem}{Theorem}
 
\newtheorem{prop}{Proposition} 
\newtheorem{cor}{Corollary}
\newtheorem{rem}{Remark}
\newtheorem{definition}{Definition}
\newtheorem*{proof*}{Proof}

\begin{document}
\newpage
\clearpage
\title{{\bf QW-Search/Zeta Correspondence}
\vspace{15mm}} 

\author{
  Taisuke HOSAKA \\
  Graduate School of Environment and Information Sciences \\
  Yokohama National University \\
  Hodogaya, Yokohama, 240-8501, Japan \\
  e-mail: hosaka-taisuke-pn@ynu.jp \\
  \\ \\
  Norio KONNO \\
  Department of Mathematical Sciences \\
  College of Science and Engineering \\
  Ritsumeikan University \\
  1-1-1 Noji-higashi, Kusatsu, 525-8577, JAPAN \\
  e-mail: n-konno@fc.ritsumei.ac.jp \\
  \\ \\
  Etsuo SEGAWA \\
  Graduate School of Environment and Information Sciences \\
  Yokohama National University \\
  Hodogaya, Yokohama, 240-8501, Japan \\
  e-mail: segawa-etsuo-tb@ynu.ac.jp \\
  \\ \\
}

\date{\empty }

\maketitle

\vspace{80mm}
\noindent


\vspace{20mm}

\par\noindent

\clearpage

\begin{abstract}
   We consider the connection between this zeta function and quantum search via quantum walk.
   First, we give an explicit expression of the zeta function on the one-dimensional torus in the general case of 
   the number and position of marked vertices.
   Moreover, we deal with the two special cases of the position of the marked vertices on the $d$-dimensional torus $(d \ge 2)$.
   Additionally, we treat the property of the zeta function by using the Mahler measure.
   Our results show the relationship between the zeta function and quantum search algorithms for the first time.
\end{abstract}
  
\vspace{10mm}
  
\begin{small}
\end{small}
\vspace{10mm}

\section{\bf Introduction \label{sec: intro}}
Quantum walk (QW) has been recently investigated as a counterpart of classical random walk (RW).
QW has interested properties compared to RW like linear spreading and localization, see \cite{IKS,KLS}.
By these features, QW is used as a tool for quantum search problems.
On the other hand, Komatsu et al. \cite{Walk/Zeta} introduced a new type of zeta functions for various walks including RW and QW on $T_{N}^{d}$ called the walk-type zeta function,
and bridged between the zeta function and a class of walks.
Here, $T_{N}^{d}$ denotes the $d$-dimensional torus with $N^{d}$ vertices.
Such a relationship is called ``Walk/Zeta Correspondence'' in \cite{Walk/Zeta}.
The zeta correspondence is studied for various models and found a relationship to different fields of mathematics and physics.
For example, Endo et al.\,\cite{EKKS} presented the zeta function based on bipartite walk \cite{CGSZ}.
Moreover, Komatsu et al.\,\cite{Mahler/Zeta} clarified that the walk-type zeta function is related to the Mahler measure, which appeared such as number theory and dynamical system \cite{M}.

In this paper, we consider the walk-type zeta function for quantum search.
We call such a relationship ``QW-search/Zeta Correspondence'' following "Walk/Zeta Correspondence".
For $T_{N}^{1}$, we get an expression of the walk-type zeta function in general case.
Furthermore, when the number of marked vertices is fixed to half,
we obtain representations for $T_{2N}^{d}$ in the following two cases:
Case 1 is that any marked vertex and non-marked vertex are adjacent.
Case 2 is that vertices in half of $T_{2N}^{d}$ are all marked vertices.
Additionally, we show the characteristics of the zeta function for quantum search, comparing to two cases of searching the marked vertices.
Our results connect the zeta function with quantum search algorithms based on QW for the first time.
To study quantum search algorithms by using zeta function may be useful for application to quantum information theory. 

The rest of this paper is as follows.
In Section 2, we explain the model of quantum search.
Section 3 presents the walk-type zeta function with respect to quantum search on $T_{N}^{1}$.
In particular, we consider a general case of the position of the marked vertices.
Section 4 treats the model on $T_{2N}^{d}$ with half marked vertices.
In Section 5, we compare the difference in the zeta function between quantum search and non-quantum search. 
Section 6 concludes our results.

\section{\bf The model of quantum search \label{sec2}}
Let $G=(X,E)$ be a connected and simple graph with $N$ vertices and $\epsilon$ edges.
Here $X$ is the set of vertices and $E$ is the set of edges.
Moreover, let $M_{X}\subset X$ be a set of the marked vertices with $|M_{X}|=m$.
The duplication $G'=(X \sqcup Y, E')$ of $G$ is as follows: 
$Y$ is the copy of $X$ (i.e., $Y=\{x' \,|\, x \in X\}$) and $"\sqcup"$ is the disjoint union.
$E'$ is denoted by
\begin{align*}
  \{x,y\} \in E \Leftrightarrow \{x,y'\}, \{x',y\} \in E'.
\end{align*}
Let $M \subset X \sqcup Y$ be a set of the marked vertices such that
\begin{align*}
  M=M_{X} \sqcup M_{Y},
\end{align*}
where $M_{Y}=\{x' \in Y \,|\, x \in X \}$.
By using the duplication $G'$, we define modified graph $G_{M}=(X \sqcup Y, E_{M})$ as follows:
The edges set $E_{M}$ is denoted by
\begin{align*}
  E_{M}=E' \cup E_{2},
\end{align*}
where
\begin{align*}
  E_{2}=\{(x,x')\,|\, x \in M_{X}, x' \in M_{Y}\}.
\end{align*}
Note that $|E_{M}|=2 \epsilon+m$.
Then we define $(2\epsilon+m) \times N$ matrices $K$ and $L$ as
\begin{align*}
  (K)_{e,x}=
  \left\{
    \begin{array}{ll}
      1/\sqrt{\mathrm{d}_{G}(x)} & : \mbox{$x \in e$ and $x \notin M_{X}$},\\
      1 & : \mbox{$x \in e$ and $e \in E_{2}$}, \\
      0 & : \mbox{otherwise},
    \end{array}
  \right.
\end{align*}

\begin{align*}
  (L)_{e,x'}=
  \left\{
    \begin{array}{ll}
      1/\sqrt{\mathrm{d}_{G}(x')} & : \mbox{$x' \in e$ and $x' \notin M_{Y}$},\\
      1 & : \mbox{$x' \in e$ and $e \in E_{2}$}, \\
      0 & : \mbox{otherwise},
    \end{array}
  \right.
\end{align*}
where $e \in E_{M}, x \in X, x' \in Y$ and $\mathrm{d}_{G}(x)$ is the degree of $x$ in the original graph $G$.
We should remark that once a walker steps in $M$, it can't escape from $M$ forever.
By using these matrices, the time evolution matrix on $G_{M}$ is defined by
\begin{align}
  \label{eq:time-evo}
  W'=(2LL^{\top}-I_{2\epsilon+m})(2KK^{\top}-I_{2\epsilon+m}).
\end{align}
Eq.\,(\ref{eq:time-evo}) can be interpreted that if a quantum walker on an edge such that either one of the endpoints is marked, then it is reflected with phase reversal $-1$;
if a quantum walker on an edge such that both of the endpoints are marked, then it stays the same edge with weight 1;
otherwise it transitions to a superposition on the neighboring edges following the Grover's matrix.
This is the time evolution of the quantum search driven by Grover walk proposed by \cite{S}.

On the other hand, similarly to \cite{Walk/Zeta},
we define a zeta function for a time evolution matrix $U$ on $T^{d}_{N}$ as follows:
\begin{definition}[\cite{Walk/Zeta}]
  The zeta function for a time evolution matrix $U$ on $T_{N}^{d}$ is defined by
  \begin{align*}
    \zeta(U, T_{N}^{d},u)=\mathrm{det}\left(I_{2dN^{d}+m}-uU\right)^{-1/N^{d}}.
  \end{align*}
\end{definition}
Our interest is an exponential expression for the above zeta function for $U=W'$,
in particular, $N \rightarrow \infty$. 
For the time evolution matrix $W'$, in order to prove our main results Theorem \ref{Thm: QW_Zeta} and \ref{Thm:zeta_d-dim}, we will use the following fact given in \cite{KSS}.
\begin{prop}{\rm (Konno, Sato and Segawa \cite{KSS})} \\
  \label{Prop:KSS}
  Let G be a connected graph with N vertices and $\epsilon$ edges.
  $W'$ is the time evolution matrix with search algorithm on $G$.
  Then we get
  \begin{align*}
    \mathrm{det}\left(I_{2N+m}-uW'\right)
    =(1-u)^{2(\epsilon-N)+3m} \mathrm{det}\left((1+u)^{2}I_{N-m}-4uP^{2}_{M}\right).
  \end{align*}
  Here, $P_{M}$ is an $(N-m) \times (N-m)$ matrix describing $RW$ with the Dirichlet boundary condition at $M$, that is,
\begin{align*}
  (P_{M})_{v,x}=
  \left\{
    \begin{array}{ll}
      1/\mathrm{d}_{G}(x) & : \mbox{$v$ and $x$ are adjacent}, \\
      0 & : \mathrm{otherwise}
    \end{array}
  \right.
\end{align*}
for $v,x \in M$.
\end{prop}

\section{\bf One-dimensional torus \label{sec: Zeta_on_cycle}}
In this section, we consider a zeta function on $G(X, E)=T_{N}^{1}$, that is,
\begin{align*}
  X&=\{x_{j}\,|\,j= 0, 1,..., N-1\, (\mathrm{mod}\,N)\}, \\
  E&=\{(x_{j},x_{j+1})\,|\, x_{j} \in V \, (\mathrm{mod}\,N)\}.
\end{align*}
Let $M$ be a set of marked vertices decomposed into
\begin{align*}
  M=\bigsqcup_{j=1}^{\ell} M_j,
\end{align*}
where $M_{j}$ forms an isolated point or a path whose length is more than 2 on $T^{1}_{N}$.
On the other hand, we will decompose the set of non-marked vertices $X \setminus M$ into the following subsets $F$ and $F'$, that is,
\begin{align*}
  M \setminus X= F \sqcup F'.
\end{align*}
Here, $F$ is expressed as 
\begin{align*}
  F=\bigsqcup_{j=1}^{r} F_j,
\end{align*}
where $F_{j}$ forms a path whose length is more than 2 on $T^{1}_{N}$;
$F'$ is denoted by
\begin{align*}
  F'=\{x \in X \setminus M \,|\, \text{every neighbor of}\, \, x \in M\}.
\end{align*}
We define the ratios of the above subsets in the number of vertices $M, F, F', F_{j}$ by $c_{M}=|M|/N, c_{F}=|F|/N, c_{F'}=|F'|/N$ and $c_{j}=|F_{j}|/N \quad (j=1,2,...,r)$.
Then we give an expression of the zeta function, which is our first main result.

\begin{theorem}
  \label{Thm: QW_Zeta}
  Let $T^{1}_{N}$ be the one-dimensional torus.
  Let $W'$ be the time evolution matrix with search algorithm on $T_{N}^{1}$.
  Then we obtain
  \begin{align}
    &\zeta(W',T^{1}_{N},u)^{-1} \notag \\
    &=\exp \left[3c_{M}\log(1-u)+2c_{F'}\log(1+u)+ \frac{1}{N} \sum_{j=1}^{r} \sum_{k=1}^{c_{j}N} \log\left\{1-2\cos\left(\frac{2k\pi}{c_{j}N+1}\right)u+u^{2}\right\}\right], \notag\\
    &\lim_{N \rightarrow \infty}\zeta(W',T^{1}_{N},u)^{-1} \notag \\ 
    \label{eq:1-dim}
    &=\exp \left[3c_{M}\log(1-u)+2c_{F'}\log(1+u)+c_{F} \int_{0}^{2\pi} \log\left(1-2\cos\theta \cdot u+u^{2}\right) \frac{d\theta}{2\pi} \right].
  \end{align}
\end{theorem}
\begin{rem}
  Theorem \ref{Thm: QW_Zeta} implies that the zeta function depends only on the ratios of the subsets $|M|,|F|$ and $|F'|$ in the limit of $N \rightarrow \infty$, which is independent of its configuration of marked and non-marked vertices.
\end{rem}
Now we will start the proof of Theorem \ref{Thm: QW_Zeta}.
\begin{proof}
  Let $D_{n}$ be $n \times n$ matrix given by
  \begin{align}
    \label{eq:d_n}
    D_{n}=
    \begin{pmatrix}
      0    & 1     & \dots & \dots & 0                                \\
      1  & 0           & \ddots & \text{\Huge 0} &   \vdots   \\
      \vdots   & \ddots      & \ddots & \ddots    &   \vdots   \\
      \vdots   & \text{\Huge 0}   & \ddots & 0    & 1  \\
      0   &      \dots       &    \dots    &  1      & 0
    \end{pmatrix}.
  \end{align}
  It follows from Eq.\,(\ref{eq:d_n}) and the definition of $P_{M}$ that
  \begin{align*}
    P_{M}=\frac{1}{2}\bigoplus_{j=1}^{r} D_{|F_{j}|} \oplus O_{|F'|},
  \end{align*}
  where $O_{|F'|}$ is the $|F'| \times |F'|$ zero matrix.
  On the other hand, we can easily calculate the spectra of $D_{n}$ as
  \begin{align}
    \label{eq: spec_A}
    {\rm Spec}\left(D_{n}\right)=\left\{2\cos{\left(\frac{k\pi}{n+1}\right)}\,|\,k=1,2,...,n\right\}
  \end{align}
  with multiplicity 1.
  Combining Proposition \ref{Prop:KSS} with Eq. (\ref{eq: spec_A}) implies
  \begin{align*}
    &\mathrm{det}\left(I_{2N+|M|}-uW'\right) \\
    &=(1-u)^{3|M|} \mathrm{det}\left((1+u)^{2}I_{N-|M|}-4uP_{M}^{2}\right) \\
    &=(1-u)^{3|M|} \prod_{\lambda \in \mathrm{Spec}(P_{M})} \left\{(1+u)^{2}-4u\lambda^{2}\right\} \\
    &=(1-u)^{3|M|}  \left\{(1+u)^{2}-4u\cdot 0\right\}^{|F'|}  \prod_{j=1}^{r} \prod_{k=1}^{|F_{j}|}\left\{(1+u)^{2}-4u \cdot \cos^{2}\left(\frac{k\pi}{|F_{j}|+1}\right)\right\} \\
    &=(1-u)^{3c_{M}N} (1+u)^{2c_{F'}N}\prod_{j=1}^{r} \prod_{k=1}^{c_{j}N} \left\{1-2\cos\left(\frac{2k\pi}{c_{j}N+1}\right) \cdot u+u^{2}\right\}.
  \end{align*}
  Therefore, we get
  \begin{align*}
    &\zeta(W',T^{1}_{N},u)^{-1} \\
    &=\left\{\mathrm{det}\left(I_{2N+|M|}-uW'\right)\right\}^{1/N} \\
    &=\exp\left[\frac{1}{N}\log\left\{\mathrm{det}\left(I_{2N+|M|}-uW'\right)\right\}\right] \\
    &=\exp\left[3c_{M}\log(1-u)+2c_{F'}\log(1+u) +\frac{1}{N} \sum_{j=1}^{r} \sum_{k=1}^{c_{j}N} \log\left\{1-2\cos\left(\frac{2k\pi}{c_{j}N+1}\right)u+u^{2}\right\}\right].
  \end{align*}
  Thus the first claim of Theorem \ref{Thm: QW_Zeta} is finished.
  By taking a limit as $N \rightarrow \infty$ for the third term in the exponential function of the above equation, we have
  \begin{align*}
    &\lim_{N \rightarrow \infty} \frac{1}{N} \sum_{j=1}^{r} \sum_{k=1}^{c_{j}N} \log\left(1-2\cos\left(\frac{2k\pi}{c_{j}N+1}\right) \cdot u+u^{2}\right) \\
    &=\sum_{j=1}^{r} \int_{0}^{c_{j}} \log\left(1-2\cos\left(\frac{2\pi x}{c_{j}}\right) \cdot u+u^{2}\right) dx \\
    &=\sum_{j=1}^{r} \int_{0}^{2\pi} \log\left(1-2\cos\theta \cdot u+u^{2}\right) \frac{c_{j}}{2\pi} d\theta \\
    &=c_{F} \int_{0}^{2\pi} \log\left(1-2\cos\theta \cdot u+u^{2}\right) \frac{d\theta}{2\pi}. 
  \end{align*}
  Hence the proof of Theorem \ref{Thm: QW_Zeta} is completed.
\end{proof}

\section{\bf $d$-dimensional torus}
From now on, we consider $T^{d}_{2N}$ with half marked vertices (i.e., $|M|=(2N)^{d}/2$).
In particular, the following two cases of the marked vertex positions are considered:
\begin{itemize}
  \item Case 1: (Marked and non-marked vertices are arranged in a checkerboard pattern)
  \begin{align*}
    M=\{(x_{1}, x_{2},...,x_{d}) \in T_{2N}^{d} \,|\, x_{1}+x_{2}+...+x_{d} \in \text{even} \}.
  \end{align*}
  \item Case 2: (The marked vertices are arranged in the half-area of $T_{2N}^{d}$)
  \begin{align*}
    M=\{(x_{1}, x_{2},...,x_{d}) \in T_{2N}^{d} \,|\, 0 \leq x_{d} \leq N \}.
  \end{align*}
\end{itemize}
We define the time evolution matrices of Cases 1 and 2 as $W'_{1}$ and $W'_{2}$, respectively.
Then we have the second main results.

\begin{theorem}
  \label{Thm:zeta_d-dim}
  Let $W'_{1}$ and $W'_{2}$ be the time evolution matrices for Case 1 and Case 2, respectively.
  Then we obtain the following equations.
  \begin{align}
    \label{eq:case_1_d-dim}
    &\zeta(W'_{1},T^{d}_{2N},u)^{-1}=\lim_{N \rightarrow \infty}\zeta(W'_{1},T^{d}_{2N},u)^{-1}=\exp\left[\left(2d-\frac{1}{2}\right)\log(1-u)+\log(1+u)\right], \\
    \label{eq:case_2}
    &\zeta(W'_{2},T^{d}_{2N},u)^{-1}=\exp\Bigg[\left(2d-\frac{1}{2}\right)\log(1-u) \Big. \notag \\
    &\Bigg. +\frac{1}{(2N)^{d}} \sum_{k_{1},...,k_{d-1}=1}^{2N} \sum_{k_{d}=1}^{N} \log\left\{(1+u)^{2}-\frac{4u}{d^{2}}\left(\sum_{i=1}^{d-1}\cos\left(\frac{(k_{i}-1)\pi}{N}\right)+\cos\left(\frac{k_{d}\pi}{N+1}\right)\right)^{2}\right\}\Bigg], \\
    \label{eq:case_2_d-dim}
    &\lim_{N \rightarrow \infty}\zeta(W'_{2},T^{d}_{2N},u)^{-1} 
    =\exp\left[\left(2d-\frac{1}{2}\right)\log(1-u) \right. \notag \\
    & \qquad +\left.\frac{1}{2} \int_{[0,2\pi)^{d}} \log\left\{(1+u)^{2}-\frac{4u}{d^{2}}\left(\sum_{i=1}^{d-1}\cos \theta_{i} +\cos \frac{\theta_{d}}{2}\right)^{2}\right\} d\Theta^{(d)}_{unif} \right],
  \end{align}
  where $\Theta^{(d)}=(\theta_{1},\theta_{2},...,\theta_{d}) \in [0,2\pi)^{d}$ and $d\Theta^{(d)}_{unif}$ is the uniform measure expressed as
  \begin{align*}
    d\Theta^{(d)}_{unif}= \frac{d\theta_{1}}{2\pi} \frac{d\theta_{2}}{2\pi} \cdot \cdot \cdot \frac{d\theta_{d}}{2\pi}.
  \end{align*}
\end{theorem}
\begin{rem}
  We should remark that if $d=1$, ``Eq.\,(\ref{eq:1-dim}) with $(c_{M},c_{F},c_{F'})=(1/2,0,1/2)$'' and ``Eq.\,(\ref{eq:case_1_d-dim})'' are the same form.
In addition, ``Eq.\,(\ref{eq:1-dim}) with $(c_{M},c_{F},c_{F'})=(1/2,1/2,0)$'' and ``Eq.\,(\ref{eq:case_2_d-dim})'' are equal.
Therefore, we see that Eq.\,(\ref{eq:1-dim}) is the general case of the parameters $(c_{M},c_{F},c_{F'})$ for $d=1$.
\end{rem}
\begin{rem}
  Note that $\zeta(W'_{1},T^{d}_{2N},u)^{-1}$ does not depend on $N$.
  Then we see the first equality in Eq.\,(\ref{eq:case_1_d-dim}) holds, that is,
  \begin{align*}
    \zeta(W'_{1},T^{d}_{2N},u)^{-1}=\lim_{N \rightarrow \infty}\zeta(W'_{1},T^{d}_{2N},u)^{-1}.
  \end{align*}
\end{rem}
\begin{rem}
  The setting of Case 2, that is, cutting the half region of the torus, causes the additional terms of ``\,$\cos\left(k_{d}\pi/N+1\right)$'' and ``\,$\cos(\theta_{d}/2)$'' in Eqs.\,(\ref{eq:case_2}) and (\ref{eq:case_2_d-dim}), respectively, comparing with the case for all the non-marked vertices in $T_{2N}^{d}$, see (\ref{eq:zeta_Mahler_non-search}).
  Section 5 is devoted to the comparion in more detail.
\end{rem}
Now we will start the proof of Theorem \ref{Thm:zeta_d-dim}.
\begin{proof}
  In Case 1, $P_{M}$ becomes the zero matrix.
  Then we get
  \begin{align*}
    \zeta(W'_{1},T^{d}_{\infty},u)^{-1} 
    &=\left\{\mathrm{det}\left(I_{2d(2N)^{d}+|M|}-uW'_{1}\right)\right\}^{1/(2N)^{d}} \\
    &=\left\{(1-u)^{2(\epsilon-N)+3|M|}\cdot  \mathrm{det}\left((1+u)^{2}I_{(2N)^{d}-|M|}-4uP^{2}_{M}\right)\right\}^{1/(2N)^{d}} \\
    &=\left[(1-u)^{(2N)^{d}(2d-1/2)} \cdot \left\{(1-u)^{2}\right\}^{2^{d-1}N^{d}}\right]^{1/(2N)^{d}} \\
    &=(1-u)^{(2d-1/2)}(1+u).
  \end{align*}
  Thus, Eq.\,(\ref{eq:case_1_d-dim}) holds.
  In Case 2, let $A(T_{2N}^{d})$ be the $(2N)^{d} \times (2N)^{d}$ adjacency matrix of $T_{2N}^d$.
  Then the $2^{d-1}N^{d} \times 2^{d-1}N^{d}$ matrix $P_{M}$ is described by
  \begin{align}
    \label{eq:pM}
    P_{M}&= \frac{1}{2d}
    \begin{pmatrix}
      A(T_{2N}^{d-1}) & I_{(2N)^{d-1}} & \dots & \dots & O \\
      I_{(2N)^{d-1}} & A(T_{2N}^{d-1}) & \ddots & \text{\Huge O} & \vdots \\
      \vdots & \ddots & \ddots & \ddots  & \vdots\\
      \vdots & \text{\Huge O}   & \ddots & \ddots    & I_{(2N)^{d-1}}  \\
      O & \dots & & I_{(2N)^{d-1}} & A(T_{2N}^{d-1})
    \end{pmatrix} \notag \\
    &=\frac{1}{2d}\left(I_{N} \otimes A(T_{2N}^{d-1}) + D_{N} \otimes I_{(2N)^{d-1}}\right).
  \end{align}
  It is known that 
  \begin{align}
    \label{eq:spec_A_d}
    \mathrm{Spec}\left(A(T_{2N}^{d})\right)=\left\{2\sum_{j=1}^{d}\cos \left(\frac{(k_{j}-1)\pi}{N}\right)\,|\,k_{1},k_{2},...,k_{d}=1,2,...,2N\right\}.
  \end{align}
  By Eqs.\,(\ref{eq:pM}) and (\ref{eq:spec_A_d}), we see \cite{AW}
  \begin{align*}
    &\mathrm{Spec}(P_{M}) \\
    &=\left\{\frac{1}{d}\left(\sum_{j=1}^{d-1}\cos\left(\frac{(k_{j}-1)\pi}{N}\right)+\cos\left(\frac{k_{d}\pi}{N+1}\right)\right) \,|\, k_1,...,k_{d-1}=1,2,...,2N,\, k_{d}=1,2,...,N\right\}
  \end{align*}
  with multiplicity $1$.
  Combining Proposition \ref{Prop:KSS} with the spectra of $P_{M}$ gives
  \begin{align*}
    &\mathrm{det}\left(I_{(2N)^{d}+|M|}-uW'_{2}\right) \\
    &=(1-u)^{(2N)^{d}(2d-1/2)} \mathrm{det}\left((1+u)^{2}I_{2^{d-1}N^{d}}-4uP_{M}^{2}\right) \\
    &=(1-u)^{(2N)^{d}(2d-1/2)} \prod_{\lambda \in \mathrm{Spec}(P_{M})}\left\{(1+u)^{2}-4u\lambda^{2}\right\} \\
    &=(1-u)^{(2N)^{d}(2d-1/2)} \\
    & \qquad \times \prod_{k_{1},...,k_{d-1}=1}^{2N} \prod_{k_{d}=1}^{N} \left[(1+u)^{2}-\frac{4u}{d^{2}}\left\{\sum_{i=1}^{d-1}\cos\left(\frac{(k_{i}-1)\pi}{N}\right)+\cos\left(\frac{k_{d}\pi}{N+1}\right)\right\}^{2}\right]. 
  \end{align*}
  Therefore, we have
  \begin{align*}
    &\zeta(W'_{2},T^{d}_{2N},u)^{-1} \\
    &=\left[\mathrm{det} \left(I_{(2N)^{d}+|M|}-uW'_{2}\right)\right]^{1/(2N)^{d}} \\
    &=\exp \left[\frac{1}{(2N)^{d}}\log \left\{\mathrm{det}\left(I_{(2N)^{d}+|M|}-uW'_{2}\right) \right\}\right] \\
    &=\exp\Bigg[(2d-1/2)\log(1-u) \Big. \\
    &\quad \Bigg. +\frac{1}{(2N)^{d}} \sum_{k_{1},...,k_{d-1}=1}^{2N} \sum_{k_{d}=1}^{N} \log\left\{(1+u)^{2}-\frac{4u}{d^{2}}\left(\sum_{i=1}^{d-1}\cos\left(\frac{(k_{i}-1)\pi}{N}\right)+\cos\left(\frac{k_{d}\pi}{N+1}\right)\right)^{2}\right\}\Bigg].
  \end{align*}
  Thus, the proof of Eq.\,(\ref{eq:case_2}) is completed.
  For the second term in the exponential function of Eq.\,(\ref{eq:case_2}), we take a limit as $N \rightarrow \infty$.
  Then we see
  \begin{align*}
    &\lim_{N \rightarrow \infty} \frac{1}{(2N)^{d}} \sum_{k_{1},...,k_{d-1}=1}^{2N} \sum_{k_{d}=1}^{N} \log\left\{(1+u)^{2}-\frac{4u}{d^{2}}\left(\sum_{i=1}^{d-1}\cos\left(\frac{(k_{i}-1)\pi}{N}\right)+\cos\left(\frac{k_{d}\pi}{N+1}\right)\right)^{2}\right\} \\
    &=\frac{1}{2^{d}} \int_{(0,2\pi]^{d-1}} \int_{0}^{\pi} \log\left((1+u)^{2}-\frac{4u}{d^{2}}\left(\sum_{i=1}^{d-1}\cos\left(\pi x_{i}\right)+\cos\left(\pi x_{d}\right)\right)^{2}\right) dx_{1} \cdot \cdot \cdot dx_{d} \\
    &=\frac{1}{2} \int_{(0,2\pi]^{d}} \log\left((1+u)^{2}-\frac{4u}{d^{2}}\left(\sum_{i=1}^{d-1}\cos \theta_{i} +\cos \frac{\theta_{d}}{2}\right)^{2}\right) \frac{d\theta_1}{2\pi} \cdot \cdot \cdot \frac{d\theta_{d}}{2\pi}.
  \end{align*}
  Hence, we get the desired conclusion.
\end{proof}

\section{\bf Comparison between search and non-search cases}
In this section, we consider the effect from the existence of search algorithm.
Let $W$ be the time evolution matrix without searching the marked vertices.
Then Endo et al. \cite{EKKS} showed an expression of the zeta function as follows:
\begin{align}
  \label{eq:d-dim_non-mark}
  &\lim_{N \rightarrow \infty} \zeta(W, T^{d}_{2N}, u)^{-1} \notag \\
  &=\exp \left[d \log(1-u)+ \frac{1}{2} \int_{[0,2\pi)^{d}} \log\left((1+u)^{2}-\frac{4u}{d^{2}}\left(\sum_{j=1}^{d}\cos \theta_{j}\right)^{2}\right) d\Theta^{(d)}_{unif}\right].
\end{align}
Comparing Eq.\,(\ref{eq:case_2_d-dim}) with Eq.\,(\ref{eq:d-dim_non-mark}), wee see that the coefficient of the first term in the exponential function of Eq.\,(\ref{eq:case_2_d-dim}) is larger than that of Eq.\,(\ref{eq:d-dim_non-mark}).
Additionally, the integrand in the exponential function has a slight difference.
This difference is caused by the existence of the search algorithm.
In order to investigate the difference between two zeta functions, we consider the correspondence between the Mahler measure and the zeta function.
For details of the definition of the Mahler measure, refer to \cite{BZ,M,MS}.
Similarly to the previous work \cite{Mahler/Zeta}, we introduce the logarithmic zeta function for the time evolution matrix $U$ defined by
\begin{align*}
  \mathcal{L}(U,T^{d}_{\infty},u)=\log \left[\lim_{N \rightarrow \infty}\left\{ \zeta(U,T^{d}_{\infty}, u)^{-1}\right\} \right].
\end{align*}
Then $\mathcal{L}(U,T^{d}_{\infty},u)$ can be expressed in terms of the Mahler measure in the following way.
\begin{cor}
  \label{cor:Mahler}

  Let $f$ be a Laurent polynomial of $X_{1},X_{2},...,X_{d}$, and $m(f)$ be the Mahler measure of $f$. 
  Let $W$ and $W'_{2}$ both be a time evolution matrices, that is, $W$ does not search the marked vertices and $W'_{2}$ does.
  Then we obtain
  \begin{align}
    &\mathcal{L}(W,T^{d}_{\infty},u)=d\log(1-u)+\log\left(-\frac{\sqrt{u}}{d}\right) \notag \\
    \label{eq:zeta_Mahler_non-search}
    &\qquad \qquad +\frac{1}{2} \cdot  m\left(\left(\sum_{j=1}^{d}(X_{j}+X_{j}^{-1})\right)^{2}-d^{2}(u+u^{-1}+2)\right), \\
    \label{eq:zeta_Mahler_search}
    &\mathcal{L}(W'_{2}, T^{d}_{\infty},u)=(2d-1/2)\log(1-u)+\log\left(-\frac{\sqrt{u}}{d}\right) \notag \\ 
    &\qquad +\frac{1}{2}\cdot m\left(\left(\sum_{j=1}^{d-1}X_{j}+X_{j}^{-1} +\sqrt{X_{d}}+\sqrt{X_{d}}^{-1}\right)^{2}-d^{2}(u+u^{-1}+2)\right).
  \end{align}
\end{cor}
\begin{proof}
  As for $W$ case, by Eq.\,(\ref{eq:d-dim_non-mark}), we get
  \begin{align*}
    &\mathcal{L}(W,T^{d}_{2N},u) \\
    &=d\log(1-u)+\frac{1}{2} \int_{[0,2\pi)^{d}} \log\left((1+u)^{2}-\frac{4u}{d^{2}}\left(\sum_{j=1}^{d}\cos \theta_{j}\right)^{2}\right) d\Theta^{(d)}_{unif} \\
    &=d\log(1-u)+\frac{1}{2}\int_{(0,2\pi]^{d}}\log\left((1+u)^{2}-\frac{u}{d^{2}}\left(\sum_{j=1}^{d} e^{i \theta_{j}}+e^{-i \theta_{j}}\right)^{2}\right) d\Theta^{(d)}_{unif} \\
    &=d\log(1-u)+\frac{1}{2}\log\left(-\frac{u}{d^{2}}\right) \\
    &+\frac{1}{2}\int_{(0,2\pi]^{d}} \log\left(\left(\sum_{j=1}^{d} e^{i \theta_{j}}+e^{-i \theta_{j}}\right)^{2}-d^{2}(u+u^{-1}+2)\right) d\Theta^{(d)}_{unif} \\
    &=d\log(1-u)+\log\left(-\frac{\sqrt{u}}{d}\right)+\frac{1}{2} \cdot m\left(\left(\sum_{j=1}^{d}(X+X^{-1})\right)^{2}-d^{2}(u+u^{-1}+2)\right).
  \end{align*}
  Thus, Eq.\,(\ref{eq:zeta_Mahler_non-search}) holds.
  On the other hand, concerning $W'_{2}$ case, it follows from Eq.\,(\ref{eq:case_2_d-dim}) that
  \begin{align*}
    &\mathcal{L}(W'_{2},T^{d}_{\infty},u) \\
    &=\frac{1}{2}(4d-1)\log(1-u)+\frac{1}{2} \int_{[0,2\pi)^{d}} \log\left((1+u)^{2}-\frac{4u}{d^{2}}\left(\sum_{i=1}^{d-1}\cos \theta_{i} +\cos \frac{\theta_{d}}{2}\right)^{2}\right) d\Theta^{(d)}_{unif} \\
    &=(2d-1/2)\log(1-u)+\log\left(-\frac{\sqrt{u}}{d}\right) \\ &+\frac{1}{2}\int_{(0,2\pi]^{d}} \log\left(\left(\sum_{j=1}^{d-1} e^{i \theta_{j}}+e^{-i \theta_{j}} +e^{i \theta_{d}/2}+e^{-i \theta_{d}/2}\right)^{2}-d^{2}(u+u^{-1}+2)\right) d\Theta^{(d)}_{unif} \\
    &=(2d-1/2)\log(1-u)+\log\left(-\frac{\sqrt{u}}{d}\right) \\ &+\frac{1}{2}\cdot m\left(\left(\sum_{j=1}^{d-1}X_{j}+X_{j}^{-1} +\sqrt{X_{d}}+\sqrt{X_{d}}^{-1}\right)^{2}-d^{2}(u+u^{-1}+2)\right).
  \end{align*}
  Therefore, a proof of Eq.\,(\ref{eq:zeta_Mahler_search}) is completed.
\end{proof}

Corollary \ref{cor:Mahler} implies that the difference between search and non-search cases is clarified by using the Mahler measure.
Note that $\sqrt{X_{d}}+\sqrt{X_{d}}^{-1}$ in Eq.\,(\ref{eq:zeta_Mahler_search}) is a formal expression, because the Mahler measures is a measure with respect to a Laurent polynomial.
Figure 1 shows the difference between $\mathcal{L}(T^{d}_{\infty},u)$ and $\mathcal{L}(W'_{2},T^{d}_{\infty},u)$ for $u \in (0,1)$.
When $u$ is close to 0, both graph are almost equal. However, the greater the value of $u$, the greater the difference between two. 

\begin{figure}[htbp]
  \begin{center}
    \label{fig:plot}
    \includegraphics[scale=.7]{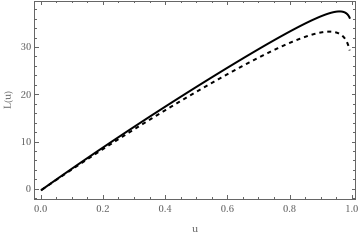}
    \caption{The solid and dot curves correspond to $\mathcal{L}(T^{d}_{\infty},u)$ and $\mathcal{L}(W'_{2},T^{d}_{\infty},u)$, respectively.}
  \end{center}
\end{figure}

\section{\bf Conclusion}
The present paper proposed a new relationship between the walk-type zeta function and quantum search based on QW for the $d$-dimensional torus $T_{N}^{d}$.
In particular, for $T_{N}^{1}$, we considered the general case for the number and position of the marked vertices.
Moreover, we treated the special two cases of the position of the marked vertices for $T_{2N}^{d}$ with the half marked vertices.
Additionally, we discussed the effect of the quantum search algorithm on the zeta function by using the Mahler measure.
One of the future problems is to get an explicit expression for the general case on $T_{N}^{d}$.
To clarify the relation between continuous time model of the quantum search algorithms and the zeta function is another interesting future problem.

\section*{Conflict of interest}
On behalf of all authors, the corresponding author states that there is no conflict of interest.

\section*{Data Availability}
Our manuscript has no associated data.


\begin{thebibliography}{}

  \bibitem{AAKV} 
  Aharonov, D., Ambainis, A., Kempe, J., Vazirani, U.:
  Quantum walks on graphs.
  Processings of the 33rd Annual ACM Symposium on Theory of Computing
  50-59 (2001)

  \bibitem{A} 
  Ambainis, A.:
  Quantum walk algorithm for element distinctness.
  SIAM J. Comput.
  {\bf 37}, 210-239 (2007)

  \bibitem{AKR} 
  Ambainis, A., Kempe J., Rivosh A.:
  Coins make quantum walks faster.
  Proceedings of the 16th ACM-SIAM Symposium on Discrete Algorithms
  1099-1108 (2005)

  \bibitem{AW}
  Brouwer, E. A., Haemers, H. W.: Spectra of Graphs, Springer, New York (2012)

  \bibitem{BZ}
  Brunault, F., Zudilin, W.:
  Many Variations of Mahler Measures: A Lasting Symphony (Australian Mathematical Society Lecture Series, Series Number 28).
  Cambridge University Press (2020) 


  \bibitem{CGSZ}
  Chen, Q., Godsil C., Sobchuk M., Zhan, H.:
  Hamiltonians of Bipartite Walks.
  arXiv:2207.01673.

  \bibitem{EKKS}
  Endo, T., Komatsu, T., Konno, N., Sato, I.:
  The spectra of the time evolution matrix of a bipartite walk on a bipartite graph.
  in preparation.

  \bibitem{IKS}
  Inui, N., Konno, N., Segawa, E.:
  One-dimensional three-state quantum walk.
  Phys. Rev. E.
  {\bf 72}, 056112 (2005)

  \bibitem{KLS}
  Konno, N., Luczak, T., Segawa, E.:
  Limit measures of inhomogeneous discrete-time quantum walks in one dimension.
  Quantum Inf. Process.
  {\bf 12}, 33-53 (2013)

  \bibitem{KSS}
  Konno, N., Sato, I., Segawa, E.:
  The spectra of the unitary matrix of a 2-tessellable staggered quantum walk on a graph.
  Yokohama Math. J. 
  {\bf 62}, 51-87 (2016)
  
  \bibitem{Grover/Zeta} 
  Komatsu, T., Konno, N., Sato, I.:
  Grover/Zeta Correspondence based on the Konno-Sato theorem, 
  Quantum Inf. Process.
  {\bf 20}, 268 (2021)

  \bibitem{Walk/Zeta} 
  Komatsu, T., Konno, N., Sato, I.:
  Walk/Zeta Correspondence.
  Journal of Statistical Physics.
  {\bf 190}, 36 (2023)

  \bibitem{CTM/Zeta}
  Komatsu, T., Konno, N., Sato, I.:
  CTM/Zeta Correspondence.
  Quantum studies: Mathematics and Foundations.
  {\bf 9}, 165-173 (2022)


  \bibitem{Mahler/Zeta} 
  Komatsu, T., Konno, N., Sato, I., Tamura, S.:
  Mahler/Zeta Correspondence.
  Quantum Inf. Process.
  {\bf 21}, 298 (2022)

  \bibitem{M} 
  Mahler, K.:
  On some inequalities for polynomials in several variables. 
  J. London Math. Soc. {\bf 37}, 341-344 (1962)

  \bibitem{MS}
  McKee, J., Smyth, C.:
  Around the Unit Circle: Mahler Measure, Integer Matrices and Roots of Unity.
  Springer (2021)

  \bibitem{P} 
  Portugal, R.: 
  Quantum Walks and Search Algorithms, 2nd edition. 
  Springer, New York (2018)

  
  \bibitem{S}
  Szegedy. M.:
  Quantum speed-up of Markov chain based algorithms.
  In Proceedings of the 45th Symposium on Foundations of Computer Science.
  32-41 (2004)

  \bibitem{SKW} 
  Shenvi, N., Kempe, J., Whaley, K. B.:
  A quantum random walk search algorithm.
  Phys. Rev. A
  {\bf 67}, 052307 (2003)
        

\end{thebibliography}
\end{document}